\documentclass[conference]{IEEEtran}

\usepackage{graphicx}
\usepackage{amsfonts}
\usepackage{amsmath}
\usepackage{amssymb}
\usepackage{color}

\newtheorem{thm}{Theorem}

\newtheorem{lem}[thm]{Lemma}
\newtheorem{proof}[thm]{proof}

\newtheorem{defn}[thm]{Definition}
\newtheorem{rem}[thm]{Remark}
\newtheorem{exam}[thm]{Example}

\ifCLASSINFOpdf

\else

\fi

\begin{document}

\title{Erasure codes with symbol locality and group decodability
for distributed storage}
\author{\IEEEauthorblockN{Wentu Song}
\IEEEauthorblockA{Singapore University of Technology\\ and Design,
Singapore\\
Email: wentu\_song@sutd.edu.sg} \and \IEEEauthorblockN{Son Hoang
Dau} \IEEEauthorblockA{Singapore University of Technology\\ and
Design,Singapore\\
Email: sonhoang\_dau@sutd.edu.sg} \and \IEEEauthorblockN{Chau
Yuen} \IEEEauthorblockA{Singapore University of Technology\\ and
Design, Singapore\\
Email: yuenchau@sutd.edu.sg}} \maketitle

\begin{abstract}
We introduce a new family of erasure codes, called \emph{group
decodable code (GDC)}, for distributed storage system. Given a set
of design parameters $\{\alpha,\beta,k,t\}$, where $k$ is the
number of information symbols, each codeword of an
$(\alpha,\beta,k,t)$-group decodable code is a $t$-tuple of
strings, called \emph{buckets}, such that each bucket is a string
of $\beta$ symbols that is a codeword of a $[\beta,\alpha]$ MDS
code (which is encoded from $\alpha$ information symbols). Such
codes have the following two properties:
\begin{itemize}
 \item [(P1)] \emph{Locally Repairable}: Each code symbol
 has locality $(\alpha,\beta-\alpha+1)$.
 \item [(P2)] \emph{Group decodable}: From each bucket we can
 decode $\alpha$ information symbols.
\end{itemize}
We establish an upper bound of the minimum distance of $(\alpha,
\beta, k, t)$-group decodable code for any given set of
$\{\alpha,\beta,k,t\}$; We also prove that the bound is achievable
when the coding field $\mathbb F$ has size $|\mathbb
F|>{n-1\choose k-1}$.
\end{abstract}


\IEEEpeerreviewmaketitle

\section{Introduction}
Distributed storage systems (DSS) are becoming increasingly
important due to the explosively grown demand for large-scale data
storage, including large files and video sharing, social networks,
and back-up systems. Distributed storage systems store a
tremendous amount of data using a massive collection of
distributed storage nodes and, to ensure reliability against node
failures, introduce a certain of redundancy.

The simplest form of redundancy is replication. DSS with
replication are very easy to implement, but extremely inefficient
in storage efficiency, incurring tremendous waste in devices and
equipment. In recent years, some efficient schemes for distributed
storage systems, such as erasure codes \cite{Weather02} and
regenerating codes \cite{Dimakis10}, are proposed. We focus on
erasure codes in this paper.

MDS codes are the most efficient erasure codes in term of storage
efficiency. When use an $[n,k]$ MDS code, the data file that need
to be stored is divided into $k$ information packets, where each
packet is a symbol of the coding field. These $k$ information
packets are encoded into $n$ packets and stored in $n$ storage
nodes such that each node stores one packet. Then the original
file can be recovered from any $k$ out of the $n$ coded packets.
Although MDS code is storage optimal, it is not efficient for node
repair. That is, when one storage node fails, we must download the
whole file from some other $k$ nodes to reconstruct the coded
packet stored in it.

To construct erasure codes with more repair efficiency than MDS
codes, the concepts of locality and locally repairable code (LRC)
were introduced \cite{Gopalan12,Prakash12,Papail122}. Let
$1\leq\alpha\leq k$ and $\delta\geq 2$. The $i$th code symbol
$c_i~(1\leq i\leq n)$ in an $[n,k]$ linear code $\mathcal C$ is
said to have locality $(\alpha,\delta)$ if there exists a subset
$S_i\subseteq[n]=\{1,2,\cdots,n\}$ containing $i$ and of size
$|S_i|\leq\alpha+\delta-1$ such that the punctured subcode of
$\mathcal C$ to $S_i$ has minimum distance at least $\delta$. We
will call each subset $\{c_j;j\in S_i\}$ a \emph{repair group}.
Thus, if $c_i$ has locality $(\alpha,\delta)$, then $c_i$ can be
computed from any $|S_i|-\delta+1$ other symbols in the repair
group $\{c_j;j\in S_i\}$. A code is said to have all-symbol
locality $(\alpha,\delta)$ (or is called an $(\alpha,\delta)_a$
code) if all of its code symbols have locality $(\alpha,\delta)$.
Note that $|S_i|-\delta+1\leq\alpha$. The code has a higher repair
efficiency than MDS code if $\alpha<k$. The minimum distance of an
$(\alpha,\delta)_a$ linear code is bounded by (See
\cite{Prakash12}) :
\begin{align} d\leq
n-k+1-\left(\left \lceil\frac{k}{\alpha}\right \rceil-1\right )
(\delta-1).\label{eq-lrc-bound}
\end{align}
However, for the case that $(\alpha+\delta-1)\nmid n$ and
$\alpha|k$, there exists no $(\alpha,\delta)_a$ linear code
achieving the above bound \cite{Wentu14}.

The most common case of $(\alpha,\delta)_a$ linear code is that
$n$ is divisible by $\alpha+\delta-1$. For this case, in the
constructions presented in the literature, all code symbols of an
$(\alpha,\delta)_a$ linear code are usually divided into
$t=\frac{n}{\alpha+\delta-1}$ mutually disjoint repair groups such
that each repair group is a codeword of an
$[\alpha+\delta-1,\alpha]$ MDS code. Fig. \ref{fg-ex-code1}
illustrates a $(4,3)_a$ systematic linear code with $n=18$ and
$k=6$, where $x_1,\cdots,x_6$ are the information symbols and
$y_1,\cdots,y_{12}$ are the parities. All code symbols are divided
into three groups and each group is a codeword of a $[6,4]$ MDS
code. By constructing the parities elaborately, the code can be
distance optimal according to \eqref{eq-lrc-bound}.

\renewcommand\figurename{Fig}
\begin{figure}[htbp]
\begin{center}
\includegraphics[height=2.45cm]{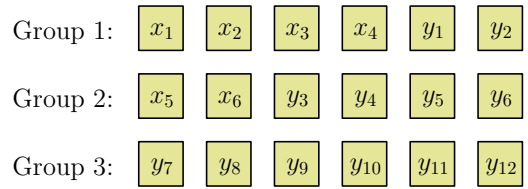}
\end{center}
\vspace{-0.2cm}\caption{Illustration of a systematic locally
repairable code: The information symbols $x_1,\cdots,x_6$ are
encoded into $x_1,\cdots,x_6,y_1,\cdots,y_{12}$ that are divided
into three groups. Each group is a codeword of a $[6,4]$ MDS code.
}\label{fg-ex-code1}
\end{figure}

As pointed out in \cite{Tamo-IT}, in distributed storage
applications there are subsets of the data that are accessed more
often than the remaining contents (they are termed ``hot data'').
Thus, a desired property of a distributed storage system is that
the subsets of hot data can be retrieved easily and by multiple
ways. For example, for the storage system illustrated by Fig.
\ref{fg-ex-code1}, suppose $x_1$ is hot data. There are two ``easy
ways'' to retrieve it from the system: Downloaded $x_1$ directly
from the node where it is stored, or decode it from any four coded
symbols in the first group. Another way is to decode it from some
six coded symbols, but this is not an easy way because to decode
$x_1$, one has to decode the whole data file.

\renewcommand\figurename{Fig}
\begin{figure}[htbp]
\begin{center}
\includegraphics[height=2.45cm]{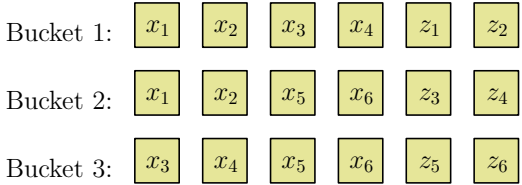}
\end{center}
\vspace{-0.2cm}\caption{Illustration of a $(4,6,6,3)$-group
decodable code: $x_1,\cdots,x_6$ are information symbols and
$z_1,\cdots,z_6$ are parities. Each codeword has $3$ buckets and
each bucket is a codeword of a $[6,4]$ MDS code that is encoded
from $4$ information symbols. Clearly, each bucket is a repair
group.}\label{fg-ex-code2}
\end{figure}

In this work, we introduce a new family of erasure codes, called
\emph{group decodable code (GDC)}, for distributed storage system,
which can provide more options of easy ways to retrieve each
information symbol than systematic codes. Given a set of design
parameters $\{\alpha,\beta,k,t\}$, where $k$ is the number of
information symbols, each codeword of an
$(\alpha,\beta,k,t)$-group decodable code is a $t$-tuple of
strings, called \emph{buckets}, such that each bucket is a string
of $\beta$ symbols that is a codeword of a $[\beta,\alpha]$ MDS
code (which is encoded from $\alpha$ information symbols). So such
codes have the following two properties:
\begin{itemize}
 \item [(P1)] \emph{Locally Repairable}: Each code symbol
 has locality $(\alpha,\beta-\alpha+1)$.
 \item [(P2)] \emph{Group decodable}: From each bucket we can
 decode $\alpha$ information symbols.
\end{itemize}

Fig. \ref{fg-ex-code2} illustrates a $(4,6,6,3)$-group decodable
code. There are six information symbols $x_1,\cdots,x_6$. Each
codeword has $3$ buckets and each bucket is a codeword of a
$[6,4]$ MDS code that is encoded from $4$ information symbols.
Clearly, each bucket is a repair group. So each code symbol of
this code has locality $(4,3)$. Moreover, this code provides more
options of easy ways to retrieve each information symbol than the
code in Fig. \ref{fg-ex-code1}. For example, $x_1$ can be
downloaded directly from two nodes or can be decoded from any four
symbols in bucket 1 or any four symbols in bucket 2. In the case
that $x_1$ is requested simultaneously by many users of the
system, the can ensure that multiple read requests can be
satisfied concurrently and with no delays.

\subsection{Our contribution}
We establish an upper bound of the minimum distance of group
decodable code for any given set of parameters
$\{\alpha,\beta,k,t\}$ (Theorem \ref{up-bnd}). We also prove that
there exist linear codes of which the minimum distances achieve
the bound, which proves the tightness of the bound (Theorem
\ref{up-bnd-tight}). Our proof gives a method to construct
$(\alpha,\beta,k,t)$-group decodable code on a field of size
$q>{n-1\choose k-1}$, where $n=t\beta$ is the length of the code.

\subsection{Related Work}
Some existing works consider erasure codes for distributed storage
that can provide multiple alternatives for repairing information
symbols or all code symbols with locality.

In \cite{Pamies13}, the authors introduced the metric ``local
repair tolerance'' to measure the maximum number of erasures that
do not compromise local repair. They also presented a class of
locally repairable codes, named pg-BLRC codes, with high local
repair tolerance and low repair locality. However, they did not
present any bound on the minimum distance of such codes.

In \cite{Wang14}, the concept of $(\alpha,\delta)_c$-locality was
defined, which captures the property that there exist $\delta-1$
pairwise disjoint local repair sets for a code symbol. An upper
bound on the minimum distance for $[n,k]$ linear codes with
information $(\alpha,\delta)_c$-locality was derived, and codes
that attain this bound was constructed for the length $n\geq
k(\alpha(\delta-1)+1)$. However, for $n< k(\alpha(\delta-1)+1)$,
it is not known whether there exist codes attaining this bound.
Upper bounds on the rate and minimum distance of codes with
all-symbol $(\alpha,\delta)_c$-locality was proved in
\cite{Tamo14}. However, no explicit construction of codes that
achieve this bound was presented. It is still an open question
whether the distance bound in \cite{Tamo14} is achievable.

Another subclass of LRC is codes with $(r,t)$-locality: any set of
$t$ code symbols are functions of at most $r$ other code symbols
\cite{Rawat-14}. Hence, for such codes, any $t$ failed code
symbols can be repaired by contacting at most $r$ other code
symbols. An upper bound of the minimum distance of such codes
similar to \eqref{eq-lrc-bound} is derived in \cite{Rawat-14}.

\subsection{Organization}
The paper is organized as follows. In Section
\uppercase\expandafter{\romannumeral 2}, we present the related
concepts and the main results of this paper. We prove the main
results in section \uppercase\expandafter{\romannumeral 3} and
Section \uppercase\expandafter{\romannumeral 4}.

\section{Model and Main Result}
Denote $[n]:=\{1,\cdots,n\}$ for any given positive integer $n$.
Let $\mathbb F$ be a finite field and $k$ be a positive integer.
For any $S=\{i_1,\cdots,i_{\alpha}\}\subseteq[k]$, the
\emph{projection} of $\mathbb F^k$ about $S$ is a function
$\psi_{S}:\mathbb F^{k}\rightarrow\mathbb F^{\alpha}$ such that
for any $(x_1,\cdots,x_k)\in\mathbb F^{k}$,
\begin{align}\label{proj}
\psi_{S}(x_1,\cdots,x_k)=(x_{i_1},\cdots,x_{i_{\alpha}}).
\end{align}

We can define group decodable codes (GDC) as follows.
\begin{defn}\label{gdc} Suppose $\mathcal
S=\{S_1,\cdots, S_t\}$ is a collection of subsets of $[k]$ and
$\mathcal N=\{n_1,\cdots, n_t\}$ is a collection of positive
integers such that $\textstyle\bigcup_{i=1}^tS_i=[k]$ and $n_i>
k_i=|S_i|, \forall i\in[t]$. A linear code $\mathcal C$ is said to
be an $(\mathcal N,\mathcal S)$-\emph{group decodable code (GDC)}
if $\mathcal C$ has an encoding function $f$ of the following
form:
\begin{eqnarray}\label{gdc-ec-fun}
f:\mathbb F^k&\rightarrow& ~ ~ ~ ~ ~ ~ ~ ~ \mathbb
F^{n_1}\times\cdots\times\mathbb F^{n_t}\nonumber\\
x ~~&\mapsto&(f_1(\psi_{S_1}(x)),\cdots, f_t(\psi_{S_t}(x))).
\end{eqnarray}
where each $f_i:\mathbb F^{k_i}\rightarrow\mathbb F^{n_i}$ is an
encoding function of an $[n_i,k_i]$ MDS code and the output of it
is called a \emph{bucket}.
\end{defn}

By Definition \ref{gdc}, if $\mathcal C$ is an $(\mathcal
N,\mathcal S)$-group decodable code, then $\mathcal C$ has length
$n=\sum_{i=1}^tn_i$. For any message vector $x=(x_1,\cdots,x_k)$
and $i\in[t]$, the subset of $k_i$ messages $\{x_j; j\in S_i\}$
are encoded into a bucket of $n_i$ symbols by the function $f_i$.
A codeword of $\mathcal C$ is the concatenation of these $t$
buckets. Since $f_i$ is an encoding function of an $[n_i,k_i]$ MDS
code, each bucket is a repair group and we can decode the subset
$\{x_j; j\in S_i\}$ from any $k_i$ symbols of the $i$th
bucket.---The term ``group decodable code'' comes from this
observation.

For the special case that $S_1,\cdots,S_t$ are pairwise disjoint,
an $(\mathcal N,\mathcal S)$-group decodable code $\mathcal C$ is
just the direct sum of the $t$ buckets and the minimum distance of
$\mathcal C$ is $\min\{n_i-k_i+1;i\in[t]\}$. In this work, we
consider the most general case that $S_1,\cdots,S_t$ can have
arbitrary intersection.

Definition \ref{gdc} depends on the explicit collections $\mathcal
S$ and $\mathcal N$. We can also define GDC based on design
parameters.

\begin{defn}\label{gdc-pmt} Let $\alpha, \beta,
k, t$ be positive integers such that $\alpha<\text{min}\{k,
\beta\}$. A linear code $\mathcal C$ is said to be an $(\alpha,
\beta, k, t)$-\emph{group decodable code} if $\mathcal C$ is an
$(\mathcal N,\mathcal S)$-group decodable code for some $\mathcal
S=\{S_1,\cdots, S_t\}$ and $\mathcal N=\{n_1,\cdots, n_t\}$ such
that $S_i\subseteq[k]$, $|S_i|=\alpha$ and $n_i=\beta$ for all
$i\in[t]$.
\end{defn}

If $\mathcal C$ is an $(\alpha, \beta, k, t)$-group decodable
code, then by Definition \ref{gdc-pmt}, the length of $\mathcal C$
is $n=t\beta$. Moreover, since $\textstyle\bigcup_{i=1}^tS_i=[k]$
and $|S_i|=\alpha$, then $t\alpha=\sum_{i=1}^t|S_i|\geq k$, which
implies that $\left\lfloor\frac{t\alpha}{k}\right\rfloor\geq 1$.
So we have the following remark.
\begin{rem}\label{rem-gdc-pmt}
If $\mathcal C$ is an $(\alpha, \beta, k, t)$-group decodable
code, then $n=t\beta$ and
$\left\lfloor\frac{t\alpha}{k}\right\rfloor\geq 1$.
\end{rem}

We will give a tight upper bound on the minimum distance $d$ of an
$(\alpha, \beta, k, t)$-group decodable code $\mathcal C$. Our
main results are the following two theorems.

\begin{thm}\label{up-bnd} Let $t\alpha=sk+r$ such that $s\geq 1$
and $0\leq r\leq k-1$. If $\mathcal C$ is an $(\alpha, \beta, k,
t)$-group decodable code, then
\begin{align} d\leq s\beta-\left\lceil\frac{k-r}{{t\choose
s}}\right\rceil+1. \label{eq-up-bnd}\end{align}
\end{thm}

Note that an $(\alpha, \beta, k, t)$-group decodable code is an
$(r,\delta)_a$ with the additional property (P2). So the bound
\eqref{eq-up-bnd} is looser than the bound \eqref{eq-lrc-bound}.
The sacrifice in minimum distance is resulted from the property
(P2).

\begin{thm}\label{up-bnd-tight} If $|\mathbb F|>{n-1\choose k-1}$,
then there exists an $(\alpha, \beta, k, t)$-group decodable code
over $\mathbb F$ with $d$ achieves the bound (\ref{eq-up-bnd}).
\end{thm}


By Remark \ref{rem-gdc-pmt}, $t\alpha\geq k$. So we always have
$t\alpha=sk+r$ for some $s\geq 1$ and $0\leq r\leq k-1$. So
Theorem \ref{up-bnd} and \ref{up-bnd-tight} covers all possible
sets of parameters $\{\alpha, \beta, k,t\}$.

\section{Proof of Theorem \ref{up-bnd}}
In this section, we prove Theorem \ref{up-bnd}. We will use some
similar discussions as in \cite{Wentu12, Dau13, Dau14}.

In the rest of this paper, we always assume that $\mathcal
S=\{S_1,\cdots, S_t\}$ is a collection of subsets of $[k]$ and
$\mathcal N=\{n_1,\cdots, n_t\}$ such that
$\textstyle\bigcup_{i=1}^tS_i=[k]$ and $n_i=\beta>|S_i|=\alpha$
for all $i\in[t]$. Moreover, let $n=t\beta$ and
\begin{align}J_i=\{(i-1)\beta+1, (i-1)\beta+2, \cdots, i\beta\}.
\label{def-Ji}\end{align} Clearly, $J_1,\cdots, J_t$ are pairwise
disjoint and $\bigcup_{i=1}^t J_i=[n]$.

Let $\ell$ be any positive integers and $A$ be any $k\times\ell$
matrix. If $J\subseteq[\ell]$, we use $A_J$ to denote the
sub-matrix of $A$ formed by the columns of $A$ that are indexed by
$J$. Moreover, we will use the following notations:
\begin{itemize}
 \item [1)] For $i\in[k]$ and $j\in[\ell]$, $R_A(i)$ and $C_A(j)$
 are the support of the $i$th row and the $j$th column of $A$
 respectively. Meanwhile, $|R_A(i)|$ and $|C_A(j)|$ are called the
 weight of the $i$th row and the $j$th column of $A$
 respectively.
 \item [2)] The minimum row weight of $A$ is
 \begin{align} w_{\min}(A)=\min\limits_{i\in[k]}|R_{A}(i)|.
 \label{eq-min-wght}\end{align} The $i$th row is said to be
 minimal if $|R_A(i)|=w_{\min}(A)$.
 \item [3)] The repetition number of the $i$th row, denoted by
 $\Gamma_A(i)$, is the number of $i'\in[k]$ such that
 $R_A(i')=R_A(i)$. Let $\Phi_A$ be the set of indices of all minimal rows of
 $A$. We denote
 \begin{align} \Gamma(A)=\max\limits_{i\in \Phi_A}\Gamma_{A}(i).
 \label{eq-gama-mat}\end{align}
\end{itemize}

Clearly, we always have $\Gamma(M)\geq 1$. The following example
gives some explanation of the above notations.
\begin{exam}\label{ex-ntn}
Consider the following $7\times 8$ binary matrix
\begin{eqnarray*}
~ ~ ~ A=\left[\begin{array}{cccccccc}
1 & 0 & 1 & 0 & 0 & 0 & 1 & 0\\
0 & 1 & 0 & 1 & 0 & 1 & 0 & 0\\
0 & 0 & 1 & 0 & 1 & 0 & 1 & 1\\
1 & 0 & 0 & 0 & 1 & 1 & 0 & 1\\
0 & 1 & 0 & 1 & 0 & 1 & 0 & 0\\
1 & 0 & 1 & 0 & 0 & 0 & 1 & 0\\
0 & 1 & 0 & 1 & 1 & 0 & 0 & 1\\
\end{array}\right].
\end{eqnarray*}
We have $R_A(1)=R_A(6)$. So $\Gamma_A(1)=\Gamma_A(6)=2$.
Similarly, $\Gamma_A(2)=\Gamma_A(5)=2$ and the repetition number
of all other rows are $1$. Note that $w_{\text{min}}(A)=3$ and the
minimal rows of $A$ are indexed by $\{1,2,5,6\}$. Then
$\Gamma(A)=2$.
\end{exam}

To prove Theorem \ref{up-bnd}, we first give a description of
$(\mathcal N,\mathcal S)$-group decodable codes using their
generator matrix. To do this, we need the following two
definitions.

\begin{defn}\label{supp-m-m}
Let $M=(m_{i,j})_{k\times n}$ be a binary matrix and
$G=(a_{i,j})_{k\times n}$ be a matrix over $\mathbb F$. We say
that $G$ is supported by $M$ if for all $i\in[k]$ and $j\in[n]$,
$m_{i,j}=0$ implies $a_{i,j}=0$. If $\mathcal C$ is a linear code
over $\mathbb F$ and has a generator matrix $G$ supported by $M$,
we call $M$ a \emph{support generator matrix} of $\mathcal C$.
\end{defn}

\begin{defn}\label{ext-inc-mx}
Let $M_0$ be a $k\times t$ binary matrix and $M$ be a $k\times n$
binary matrix such that $C_{M_0}(j)=S_j$ for all $j\in[t]$ and
$C_{M}(j)=S_i$ for all $i\in[t]$ and $j\in J_i$. We call $M_0$ the
\emph{incidence matrix} of $\mathcal S$ and $M$ the
\emph{indicator matrix} of $(\mathcal N,\mathcal S)$.
\end{defn}

\begin{rem}\label{rem-inc-mat} Since
$\textstyle\bigcup_{i=1}^tS_i=[k]$ and $C_{M_0}(i)=S_i$ for all
$i\in[t]$, then  by Definition \ref{ext-inc-mx}, each row of $M_0$
has at least one $1$ and each column of $M_0$ has exactly $\alpha
~ 1$s. Moreover, by \eqref{def-Ji} and Definition
\ref{ext-inc-mx}, $M$ is extended from $M_0$ by replicating each
column of $M_0$ by $\beta$ times. Hence, each row of $M$ has at
least $\beta$ $1$s and each column of $M_0$ has exactly $\alpha ~
1$s.
\end{rem}

Now, we can describe $(\mathcal N,\mathcal S)$-group decodable
codes using their generator matrix.
\begin{lem}\label{spbm-gdc} Let $M$ be the indicator matrix of
$(\mathcal N, \mathcal S)$. Then $\mathcal C$ is an $(\mathcal
N,\mathcal S)$-group decodable code if and only if $\mathcal C$
has a generator matrix $G$ satisfying the following two
conditions:
\begin{itemize}
 \item[(1)] $G$ is supported by $M$;
 \item[(2)] $\text{rank}(G_{J})=\alpha$ for each $i\in[t]$ and
 $J\subseteq J_i$ with $|J|=\alpha$.
\end{itemize}
\end{lem}
\begin{proof}
This lemma can be directly derived from Definition \ref{gdc} and
\ref{ext-inc-mx}.
\end{proof}

For any $[n,k]$ linear code $\mathcal C$, the well-known Singleton
bound ([15, Ch1]) states that $d\leq n-k+1$. On the other hand, we
always have $d\geq 1$. So it must be that $d=n-k+1-\delta$ for
some $\delta\in\{0,1,\cdots, n-k\}$. The following lemma describes
a useful fact about $d$ for any linear code \cite{MacWilliams}.
\begin{lem}\label{lc-gdc}
Let $\mathcal C$ be an $[n,k,d]$ linear code and $G$ be a
generator matrix of $\mathcal C$. Let $0\leq\delta\leq n-k$. Then
$d\geq n-k+1-\delta$ if and only if any $k+\delta$ columns of $G$
has rank $k$.
\end{lem}

Using this lemma, we can give a bound on the minimum distance of
any linear code by its support generator matrix.
\begin{lem}\label{spbm-dst}
Let $M=(m_{i,j})$ be a $k\times n$ binary matrix and
$0\leq\delta\leq n-k$. The following three conditions are
equivalent:
\begin{itemize}
 \item[(1)] There is an $[n,k]$ linear code $\mathcal C$ over some
 field $\mathbb F$ such that $M$ is a support generator matrix of
 $\mathcal C$ and $d\geq n-k+1-\delta$.
 \item[(2)] $|\mathop{\textstyle\bigcup}_{j\in J}C_M(j)|\geq\ell$
 for any $\ell\in[k]$ and any $J\subseteq[n]$ of size $|J|=\ell+\delta$.
 \item[(3)] $|\mathop{\textstyle\bigcup}_{i\in I}R_M(i)|
 \geq n-k+|I|-\delta$ for all $\emptyset\neq I\subseteq[k]$.
\end{itemize}
Moreover, if condition (2) or (3) holds, there exists an $[n,k]$
linear code over the field of size $q>{n-1\choose k-1}$ with a
support generator matrix $M$ and minimum distance $d\geq
n-k+1-\delta$.
\end{lem}
\begin{proof}
The proof is given in Appendix A.
\end{proof}

For $(\mathcal N,\mathcal S)$-group decodable code, we have the
following two lemmas.
\begin{lem}\label{gdc-lem}
Suppose $M$ is the indicator matrix of $(\mathcal N, \mathcal S)$.
If $M$ satisfies condition (2) of Lemma \ref{spbm-dst}, there
exists an an $(\mathcal N,\mathcal S)$-group decodable code over
the field of size $q>{n-1\choose k-1}$ with minimum distance
$d\geq n-k+1-\delta$.
\end{lem}
\begin{proof}
The proof is given in Appendix B.
\end{proof}

\begin{lem}\label{min-dst-rglr} Let
$M_0=(m_{i,j})$ be the incidence matrix of $\mathcal S$. For any
$(\mathcal N,\mathcal S)$-group decodable code $\mathcal C$, we
have
\begin{align} \label{min-dst-rglr-eq} d\leq
w_{\min}(M_0)\beta-\Gamma(M_0)+1.\end{align} Moreover, there exist
an $(\mathcal N,\mathcal S)$-group decodable code over the field
of size $q>{n-1\choose k-1}$ with
$d=w_{\min}(M_0)\beta-\Gamma(M_0)+1$.
\end{lem}
\begin{proof}
The proof is given in Appendix C.
\end{proof}

Now, we can prove Theorem \ref{up-bnd}.
\begin{proof}[Proof of Theorem \ref{up-bnd}]
Suppose $\mathcal C$ is an $(\alpha, \beta, k, t)$-group decodable
code. By Definition \ref{gdc-pmt}, $\mathcal C$ is an $(\mathcal
N,\mathcal S)$-group decodable code for some $\mathcal
S=\{S_1,\cdots, S_t\}$ and $\mathcal N=\{n_1,\cdots, n_t\}$ such
that $S_i\subseteq[k]$, $|S_i|=\alpha$ and $n_i=\beta$ for all
$i\in[t]$. Let $M_0$ be the incidence matrix of $\mathcal S$. By
Lemma \ref{min-dst-rglr}, it is sufficient to prove
$w_{\min}(M_0)\beta-\Gamma(M_0)+1\leq s\beta-
\left\lceil\frac{k-r}{{t\choose s}}\right\rceil+1$.

By Remark \ref{rem-inc-mat}, each column of $M_0$ has exactly
$\alpha$ ones. Then the total number of ones in $M_0$ is
$N_{\text{one}}=t\alpha$. On the other hand, each row of $M_0$ has
at least $w_{\min}(M_0)$ ones. So $N_{\text{one}}=t\alpha\geq
kw_{\min}(M_0)$, which implies
$w_{\min}(M_0)\leq\frac{t\alpha}{k}$. Since $w_{\min}(M_0)$ is an
integer, then we have
\begin{align} w_{\min}(M_0)\leq\left\lfloor\frac{t\alpha}{k}\right\rfloor
=s.\label{eq2-up-bnd}\end{align} Note that $\Gamma(M_0)\geq 1$. If
$k-r\leq{t\choose s}$, then we have
$\left\lceil\frac{k-r}{{t\choose s}}\right\rceil=1$, and
\eqref{eq2-up-bnd} implies $w_{\min}(M_0)\beta-\Gamma(M_0)+1\leq
s\beta=s\beta-\left\lceil\frac{k-r}{{t\choose s}}\right\rceil+1$.
Thus, we only need to consider $k-r>{t\choose s}$. Again by
\eqref{eq2-up-bnd}, we have the following two cases:

Case 1: $w_{\min}(M_0)=s$. Let $N_s$ be the number of rows of
$M_0$ with weight $s$. Then $M_0$ has $k-N_s$ rows with weight at
least $s+1$. So the total number of ones in $M_0$ is
$N_{\text{one}}=t\alpha=sk+r\geq sN_s+(s+1)(k-N_s)=ks+(k-N_s)$.
Thus,
\begin{align} N_s\geq k-r.\label{eq3-up-bnd}\end{align}
If $\Gamma(M_0)<\left\lceil\frac{k-r}{{t\choose s}}\right\rceil$,
then the repetition number of each row of weight $w_{\min}(M_0)=s$
is at most $\left\lceil\frac{k-r}{{t\choose s}}\right\rceil-1$.
Note that there are at most ${t\choose s}$ binary vector of length
$t$ and weight $s$. Then we have $N_s\leq{t\choose
s}\left(\left\lceil\frac{k-r}{{t\choose
s}}\right\rceil-1\right)<k-r$, which contradicts to
\eqref{eq3-up-bnd}. So we have
$\Gamma(M_0)\geq\left\lceil\frac{k-r}{{t\choose s}}\right\rceil$.
Thus, $w_{\min}(M_0)\beta-\Gamma(M_0)+1\leq
s\beta-\left\lceil\frac{k-r}{{t\choose s}}\right\rceil+1$.

Case 2: $w_{\min}(M_0)\leq s-1$. Note that $t\alpha=sk+r\geq sk$
and $\alpha\leq k$. Then we have $t\geq s$ and ${t-1\choose
s-1}\geq 1$. Thus,
\begin{align*} k-r&\leq k\leq{t-1\choose
s-1}k+\frac{r}{s}{t-1\choose s-1}\\&=\frac{sk+r}{s}{t-1\choose
s-1}=\frac{t\alpha}{s}{t-1\choose s-1}=\alpha{t\choose s}.
\end{align*} So we have $\frac{k-r}{{t\choose
s}}\leq\alpha$, which implies that
\begin{align*} \left\lceil\frac{k-r}{{t\choose
s}}\right\rceil-1<\frac{k-r}{{t\choose
s}}\leq\alpha\leq\beta.\end{align*} Note that $\Gamma(M_0)\geq 1$.
Then $w_{\min}(M_0)\beta-\Gamma(M_0)+1\leq w_{\min}(M_0)\beta\leq
(s-1)\beta=s\beta-\beta\leq s\beta-\left\lceil\frac{k-r}{{t\choose
s}}\right\rceil+1$.

By above discussion, we proved
$w_{\min}(M_0)\beta-\Gamma(M_0)+1\leq
s\beta-\left\lceil\frac{k-r}{{t\choose s}}\right\rceil+1$. By
Lemma \ref{min-dst-rglr}, $d\leq
s\beta-\left\lceil\frac{k-r}{{t\choose s}}\right\rceil+1$.
\end{proof}

\section{Proof of Theorem \ref{up-bnd-tight}}
In this section, we prove Theorem \ref{up-bnd-tight}. We first
give a lemma that will be used in our following discussion.

\begin{lem}\label{mat-contr}
Suppose $t\alpha=sk+r$, where $s\geq 1$ and $0\leq r\leq k-1$. If
$k-r\leq{t\choose s}$, then there exists a $k\times t$ binary
matrix $M_0=(m_{i,j})$ such that: (i) Each column of $M_0$ has
exactly $\alpha ~ 1$s; (ii) $w_{min}(M_0)=s$ and $\Gamma(M_0)=1$.
\end{lem}
\begin{proof}
The proof is given in Appendix D.
\end{proof}

Now we can prove Theorem \ref{up-bnd-tight}
\begin{proof}[Proof of Theorem \ref{up-bnd-tight}]
By Lemma \ref{min-dst-rglr}, it is sufficient to construct a
$k\times t$ binary matrix $M_0$ such that each column has exactly
$\alpha ~ 1$s, $w_{\min}(M_0)=s$ and $\Gamma(M_0)=
\left\lceil\frac{k-r}{{t\choose s}}\right\rceil$. We have the
following two cases:

Case 1: $k-r\leq {t\choose s}$. Then
$\left\lceil\frac{k-r}{{t\choose s}}\right\rceil=1$ and $M_0$ can
be constructed by Lemma \ref{mat-contr}.

Case 2: $k-r>{t\choose s}$. In this case, we can assume
\begin{eqnarray}\label{proof-th4-eq1}
k-r=u{t\choose
s}+v
\end{eqnarray}
where $u\geq 1$ and $0\leq v\leq {t\choose s}-1$. Since
$t\alpha=sk+r$, then
\begin{align}\label{proof-th4-eq2}
t\left[\alpha-u{t-1\choose s-1}\right]&=t\alpha-tu{t-1\choose
s-1}\nonumber\\&=t\alpha-su{t\choose s}\nonumber\\
&=t\alpha-s(k-r-v)\nonumber\\
&=(t\alpha-sk)+s(r+v)\nonumber\\
&=r+s(r+v).
\end{align}
Let $M_1$ be a $u{t\choose s}\times t$ binary matrix such that
each binary vector of length $t$ and weight $s$ appears in $M_1$
exactly $u$ times. Then each column of $M_1$ has exactly
$u{t-1\choose s-1} ~ 1$s. We can further construct a $(r+v)\times
t$ matrix $M_2$ and let $M_0=\left[^{M_1}_{M_2}\right]$. To do so,
we need to consider the following two sub-cases:

Case 2.1: $v=0$. By \eqref{proof-th4-eq2},
$t\left[\alpha-u{t-1\choose s-1}\right]=(s+1)r$. It is easy to
construct an $r\times t$ binary matrix $M_2$ such that each column
has exactly $\alpha-u{t-1\choose s-1} ~1$s and each row has
exactly $s+1 ~1$s. Let $M_0=\left[^{M_1}_{M_2}\right]$. Then $M_0$
is a $k\times t$ binary matrix and each column has exactly $\alpha
~ 1$s. Moreover, by the construction, we have $w_{min}(M_0)=s$ and
$\Gamma(M)=u=\left\lceil\frac{k-r}{{t\choose s}}\right\rceil$.

Case 2.2: $v\neq 0$. Then $0\leq r\leq v+r-1$. Note that $0\leq
v\leq {t\choose s}-1$ and by \eqref{proof-th4-eq2},
$t\left[\alpha-u{t-1\choose s-1}\right]=s(r+v)+r$. By the same
discussion as in Lemma \ref{mat-contr}, we can construct a
$(r+v)\times t$ binary matrix $M_2$ such that: (i) Each column of
$M_2$ has exactly $\alpha-u{t-1\choose s-1} ~ 1$s; (ii)
$w_{min}(M_2)=s$ and $\Gamma(M_2)=1$. Let
$M_0=\left[^{M_1}_{M_2}\right]$. Then $M_0$ is a $k\times t$
binary matrix and each column has exactly $\alpha ~ 1$s. Moreover,
by the construction, we have $w_{min}(M_0)=s$ and
$\Gamma(M)=u+1=\left\lceil\frac{k-r}{{t\choose s}}\right\rceil$.

Thus, we can always construct a $k\times t$ binary matrix $M_0$
such that each column has exactly $\alpha ~ 1$s, $w_{\min}(M_0)=s$
and $\Gamma(M_0)=\left\lceil\frac{k-r}{{t\choose s}}\right\rceil$.
By Lemma \ref{min-dst-rglr}, there exist $(\mathcal N,\mathcal
S)$-group decodable code over the field of size $q>{n-1\choose
k-1}$ with
$d=w_{\min}(M_0)\beta-\Gamma(M_0)+1=s\beta-\left\lceil\frac{k-r}{{t\choose
s}}\right\rceil+1$.
\end{proof}

\section{Conclusions}
We introduce a new family of erasure codes, called group decodable
codes (GDC), for distributed storage systems that allows both
locally repairable and group decodable. Thus, such codes can be
viewed as a subclass of locally repairable codes (LRC). We derive
an upper bound on the minimum distance of such codes and prove
that the bound is achievable for all possible code parameters.
However, since GDC is a subclass of LRC, the minimum distance
bound of GDC is smaller than the minimum distance bound of LRC in
general.

\appendices

\section{Proof of Lemma \ref{spbm-dst}}
The proof consists of three steps: In the first step, we prove
condition (1) implies condition (2); In the second step, we prove
condition (2) implies condition (3); In the third step, we prove
that if condition (3) holds, then there exists an $[n,k]$ linear
code over the field of size $q>{n-1\choose k-1}$ with a support
generator matrix $M$ and minimum distance $d\geq n-k+1-\delta$.

\begin{proof}[Proof of Lemma \ref{spbm-dst}]
(1) $\Rightarrow$ (2). Suppose condition (1) holds. Let
$G=(a_{i,j})$ be a generator matrix of $\mathcal C$ supported by
$M$. Then for any $i\in[k]$ and $j\in[n]$, $m_{i,j}=0$ implies
$a_{i,j}=0$. Given any $\ell\in[k]$, since any $k+\delta$ columns
of $G$ has rank $k~($Lemma \ref{lc-gdc}$)$, then any $\ell+\delta$
columns of $G$ has rank at least $\ell$, i.e.,
$\text{rank}(G_J)\geq\ell$ for any $J\subseteq[n]$ of size
$|J|=\ell+\delta$. So $G_J$ has at most $k-\ell$ rows that are all
zeros, which implies $|\textstyle\bigcup_{j\in J} C_M(j)|\geq
\ell$.

(2) $\Rightarrow$ (3). We can prove this by contradiction. Suppose
$\emptyset\neq I\subseteq[k]$ and
$|\mathop{\textstyle\bigcup}_{i\in I}R_M(i)|<n-k+|I|-\delta$. Let
$J'=[n]\backslash\mathop{\textstyle\bigcup}_{i\in I}R_M(i)$. Then
$|J'|>k-|I|+\delta$ and $m_{i,j}=0$ for all $i\in I$ and $j\in
J'$. Let $\ell=k-|I|+1$ and $J\subseteq J'$ such that $|J|=\ell$.
Then $\mathop{\textstyle\bigcup}_{j\in J}C_M(j)\subseteq
[k]\backslash I$. So $|\mathop{\textstyle\bigcup}_{j\in
J}C_M(j)|\leq k-|I|=\ell-1$, which contradicts to condition (2).
Thus, it must be that $|\mathop{\textstyle\bigcup}_{i\in
I}R_M(i)|\geq n-k+|I|-\delta$.

(3) $\Rightarrow$ (1). The key is to construct a $k\times n$
matrix $G$ over a field $\mathbb F$ of size $q>{n-1\choose k-1}$
such that $G$ is supported by $M$ and any $k+\delta$ columns of
$G$ has rank $k$.

Let $X=(x_{i,j})_{k\times n}$ such that $x_{i,j}$ is an
indeterminant if $m_{i,j}=1$ and $x_{i,j}=0$ if $m_{i,j}=0$. Let
$f(\cdots,x_{i,j},\cdots)=\Pi_{P}\text{det}(P)$, where the product
is taken over all $k$ by $k$ submatrix $P$ of $X$ with
$\text{det}(P) \not\equiv \boldsymbol{O}$. Note that each
$x_{i,j}$ belongs to at most ${n-1 \choose k-1}$ submatrix $P$ and
has degree at most $1$ in each $\text{det}(P)$. Then $x_{i,j}$ has
degree at most ${n-1 \choose k-1}$ in $f(\cdots,x_{i,j},\cdots)$.
Note that $f(\cdots,x_{i,j},\cdots)=\prod_{P}\text{det}(P)
\not\equiv \boldsymbol{O}$. By [14, Lemma 4], if $|\mathbb F|>{n-1
\choose k-1}$, then there exist $a_{i,j}\in\mathbb F$ $($for $i,j$
where $m_{i,j}=1)$ such that $f(\cdots,a_{i,j},\cdots)\neq 0$. Let
$G=(a_{i,j})$~$($for $i,j$ where $m_{i,j}=0$, we set $a_{i,j}=0)$.
Then $G$ is supported by $M$. We will prove $\text{rank}(G_J)=k$
for any $J\subseteq[n]$ with $|J|=k+\delta$. By construction of
$G$, it is sufficient to prove $\text{det}(X_{J_0}) \not\equiv
\boldsymbol{O}$ for some $J_0\subseteq J$ with $|J_0|=k$.

Let $\mathcal G_J$ be the bipartite graph with vertex set $U\cup
V$, where $U=\{u_i;i\in[k]\}$, $V=\{v_{j};j\in J\}$ and $U\cap
V=\emptyset$ such that $(u_i,v_j)$ is an edge of $\mathcal G_I$ if
and only if $m_{i,j}=1$. Then for each $u_i\in U$, the set of all
neighbors of $u_i$ is $N(u_i)=\{v_j;j\in R_{M}(i)\cap J\}$. So for
all $I\subseteq [k]$, the set of all neighbors of the vertices in
$S=\{u_i;i\in I\}$ is $N(S)=\{v_j;j\in(\textstyle\bigcup_{i\in
I}R_{M}(i))\cap J\}$. By assumption,
$\left|\textstyle\bigcup_{i\in I}R_{M}(i)\right|\geq
n-k+|I|-\delta$ and $|J|=k+\delta$. So we have
$|N(S)|=\left|\textstyle\bigcup_{i\in I}R_{M}(i)\cap J\right|\geq
\left|\textstyle\bigcup_{i\in
I}R_{M}(i)\right|-\left|[n]\backslash J\right|=|I|=|S|$. By Hall's
Theorem $($[16, p. 419]$)$, $\mathcal G_J$ has a matching which
covers every vertex in $U$. Let $\mathcal M=\{(u_1,v_{\ell_1}),
\cdots, (u_k,v_{\ell_k})\}$ be such a matching and
$J_0=\{\ell_1,\cdots,\ell_k\}$. Let $\mathcal G_{J_0}$ be the
subgraph of $\mathcal G_J$ generated by $U\cup\{v_j;j\in J_0\}$.
Then $\mathcal M$ is a perfect matching of $\mathcal G_{J_0}$ and
$X_{J_0}$ is the Edmonds matrix of $\mathcal G_{J_0}$. It is well
known $($[17, p. 167]$)$ that a bipartite graph has a perfect
matching if and only if the determinant of its Edmonds matrix is
not identically zero. Hence $\text{det}(X_{J_0}) \not\equiv
\boldsymbol{O}$.

By the construction of $G$, we have $\text{det}(G_{J_0})\neq 0$
and $\text{rank}(G_J)=k$, where $J$ is any subset of $[n]$ and
$|J|=k+\delta$. Let $\mathcal C$ be the $[n,k]$ linear code
generated by $G$. By Lemma \ref{lc-gdc}, $d\geq n-k+1-\delta$.
Note that we have proved that $G$ is supported by $M$. So $M$ is a
support generator matrix of $\mathcal C$.
\end{proof}

\section{Proof of Lemma \ref{gdc-lem}}
\begin{proof}[Proof of Lemma \ref{gdc-lem}]
Let $G$ and $\mathcal C$ be constructed as in the proof of Lemma
\ref{spbm-dst}. We will prove that $\mathcal C$ is an $(\mathcal
N,\mathcal S)$-group decodable code.

By Lemma \ref{spbm-gdc}, we need to prove
$\text{rank}(G_{J})=\alpha$ for each $i\in[t]$ and each
$J\subseteq J_i$ of size $|J|=\alpha$. To prove this, it is
sufficient to construct a subset $J_0\subseteq[n]$ such that
$J\subseteq J_0$ and $\text{rank}(G_{J_0})=k$. To simplify
notations, without loss of generality, we can assume $J\subseteq
J_1$, where $J_1$ is defined by \eqref{def-Ji}. Since
$\bigcup_{i=1}^tS_i=[k]$, we can always find a collection
$\mathcal S'\subseteq\mathcal S$ (By proper naming, we can assume
$\mathcal S'=\{S_{1},S_{2},\cdots,S_{r}\}.)$ such that
$\bigcup_{i=1}^rS_{i}=[k]$ and
$I_\ell=S_{\ell}\backslash\bigcup_{i=1}^{\ell-1}
S_{i}\neq\emptyset, \ell=2,\cdots,r$. Then $\{I_1,
I_2,\cdots,I_r\}$ is a partition of $[k]$, where $I_1=S_1$. Let
$J'_1=J$ and for each $\ell\in\{2,\cdots,r\}$, pick an
$J'_\ell\subseteq J_{\ell}$ with $|J'_\ell|=|I_\ell|$. Let
$J_0=J'_1\cup J'_2\cup\cdots\cup J'_r$. Then $|J_0|=k$. Let
$\mathcal G_{J_0}$ be the bipartite graph with vertex set $U\cup
V$, where $U=\{u_i;i\in[k]\}$, $V=\{v_{j};j\in J_0\}$ and $U\cap
V=\emptyset$ such that $(u_i,v_j)$ is an edge of $\mathcal
G_{J_0}$ if and only if $m_{i,j}=1$. By Definition
\ref{ext-inc-mx}, $m_{i,j}=1$ for each $i\in I_\ell$, $j\in
J'_\ell$ and $\ell\in[r]$. So each subgraph $\mathcal
G_{I_\ell,J'_\ell}$ is a complete bipartite graph and has a
perfect matching, where $\mathcal G_{I_\ell,J'_\ell}$ is generated
by $\{u_i;i\in I_\ell\}\cup\{v_j;j\in J'_\ell\}$. So the bipartite
graph $\mathcal G_{J_0}$ has a perfect matching. By a similar
discussion as in the proof of Lemma \ref{spbm-dst},
$\text{rank}(G_{J_0})=k$. 
So $\text{rank}(G_{J})=\alpha$.

Moreover, by the proof of Lemma \ref{spbm-dst}, $G$ is supported
by $M$ and is a generator matrix of $\mathcal C$. So by Lemma
\ref{spbm-gdc}, $\mathcal C$ is an $(\mathcal N,\mathcal S)$-group
decodable code. By Lemma \ref{spbm-dst}, $d\geq n-k+1-\delta$. So
$\mathcal C$ is a code that satisfies our requirements.
\end{proof}

As an example, let $M_0$ be the matrix $A$ in Example
\ref{ex-ntn}. Then $\alpha=3,k=7$ and $t=8$. Let $\beta=5$. Then
$M$ is obtained from $M_0$ by replicating each column of $M_0$ by
$5$ times. By \eqref{def-Ji},
$J_1=\{1,\cdots,5\},\cdots,J_8=\{36,\cdots,40\}$. Let
$J=\{1,3,5\}\subseteq J_1$. We have $I_1=S_1=\{1,4,6\}$,
$I_2=S_2\backslash S_1=\{2,5,7\}$ and $I_3=S_3\backslash (S_1\cup
S_2)=\{3\}$. Moreover, we can pick $J_1'=J, J_2'=\{6,7,8\}$ and
$J_3'=\{11\}$. Then $G_{J_0}$ is of the following form:
\begin{eqnarray*}
~ ~ ~ \left[\begin{array}{ccccccc}
* & * & * & 0 & 0 & 0 & *\\
0 & 0 & 0 & * & * & * & 0\\
0 & 0 & 0 & 0 & 0 & 0 & *\\
* & * & * & 0 & 0 & 0 & 0\\
0 & 0 & 0 & * & * & * & 0\\
* & * & * & 0 & 0 & 0 & *\\
0 & 0 & 0 & * & * & * & 0\\
\end{array}\right]
\end{eqnarray*}
where stars denote the nonzero entries of $G_{J_0}$. Clearly,
$\{(1,1),(2,4),(3,6),(4,2),(5,5),(6,7),(7,3)\}$ is a perfect
matching of the corresponding bipartite graph $\mathcal G_{J_0}$.
By construction of $G$, we have $\text{det}(G_{J_0})\neq 0$ and
$\text{rank}(G_{J_0})=k=7$.

\section{Proof of Lemma \ref{min-dst-rglr}}
\begin{proof}[Proof of Lemma \ref{min-dst-rglr}]
Let $M$ be the indicator matrix of $(\mathcal N, \mathcal S)$. By
Lemma \ref{spbm-gdc}, $M$ is a support generator matrix of
$\mathcal C$. Let $\delta_0$ be the smallest number such that
$|\mathop{\textstyle\bigcup}_{j\in J}C_M(j)|\geq\ell$ for all
$\ell\in[k]$ and all $J\subseteq[n]$ of size $|J|=\ell+\delta_0$.
Then by Lemma \ref{spbm-dst}, $d\leq n-k+1-\delta_0$. By Lemma
\ref{spbm-dst} and \ref{gdc-lem}, there exists an $(\mathcal
N,\mathcal S)$-group decodable code over the field $\mathbb F$ of
size $q>{n-1\choose k-1}$ with $d=n-k+1-\delta_0$. Thus, to prove
this lemma, the key is to prove that
$\delta_0=n-w_{\min}(M_0)\beta-k+\Gamma(M_0)$.

By Definition \ref{ext-inc-mx}, $M_0$ is a $k\times t$ binary
matrix and $M$ is a $k\times n$ binary matrix such that
$C_{M_0}(i)=S_i$ for all $i\in[t]$ and $C_{M}(j)=S_i$ for all
$i\in[t]$ and $j\in J_i$. For each $\ell\in[n]$, let
\begin{align}
\xi_{M}(\ell)=\min\limits_{J\subseteq[n],|J|=\ell}
\left|\textstyle\bigcup\limits_{j\in
J}C_{M}(j)\right|.\label{M-clmn-supp}\end{align} Then by
definition of $\delta_0$, we have
\begin{align}
\delta_0=\min\{\delta; 0\leq\delta\leq
n-k,\xi_{M}(\ell+\delta)\geq\ell,
\forall\ell\in[k]\}.\label{dlt-min-ell}\end{align} For each
$i\in[t]$, let
\begin{align}
\xi_{M_0}(i)=\min\limits_{J\subseteq[n],|J|=i}
\left|\textstyle\bigcup\limits_{j\in
J}C_{M_0}(j)\right|.\label{M0-clmn-supp}\end{align} Then we have
the following four claims:\\
\noindent\textbf{Claim 1}: $\xi_{M_0}( i_0)=k-\Gamma(M_0)<
k=\xi_{M_0}( i_0+1)=\cdots=\xi_{M_0}(t)$,
where $ i_0=t-w_{\min}(M_0)$.\\
\noindent\textbf{Claim 2}: For all
$i\in[t]$ and $\ell\in J_i$, $\xi_{M}(\ell)=\xi_{M_0}(i)$.\\
\noindent\textbf{Claim 3}: $\ell'-\xi_{M}(\ell')\leq
i_0\beta-\xi_{M}( i_0\beta), ~
\forall\ell'\in[ i_0\beta]\}.$\\
\noindent\textbf{Claim 4}:
$\delta_0= i_0\beta-\xi_{M_0}( i_0)$.\\
Note that $n=t\beta$. Then Claims 1 and 4 imply that
$\delta_0=n-w_{\min}(M_0)\beta-k+\Gamma(M_0)$, which completes the
proof.
\end{proof}

\begin{proof}[Proof of Claim 1] Suppose $J\subseteq[t]$ and
$ i_0+1\leq|J|\leq t$. Then $\textstyle\bigcup_{j\in
J}C_{M_0}(j)=[k]$. Otherwise, there is an $\ell\in[k]$ such that
$\ell\notin C_{M_0}(j)$ for all $j\in J$, which implies that
$m_{\ell,j}=0$ for all $j\in J$. So
$R_{M_0}(\ell)\subseteq[t]\backslash J$ and
$|R_{M_0}(\ell)|\leq|[t]\backslash J|=t-|J|\leq
t-(i_0+1)=w_{\min}(M_0)-1$, which contradicts to
\eqref{eq-min-wght}. Thus, we proved that $\textstyle\bigcup_{j\in
J}C_{M_0}(j)=[k]$. By \eqref{M0-clmn-supp}, we have
$\xi_{M_0}(i)=k$ for $i_0+1\leq i\leq t$.

Now, suppose $J\subseteq[t]$ and $|J|= i_0=t-w_{\min}(M_0)$. We
have the following two cases:

Case 1: $J=[t]\backslash R_{M_0}(\ell)$ for some $\ell\in[k]$ such
that $|R_{M_0}(\ell)|=w_{\text{min}}(M_0)$. Then $|\bigcup_{j\in
J}R_{M_0}(j)|=k-\Gamma_{M_0}(\ell)$. This can be proved as
follows:

For each $\ell'\in[k]$ such that $R_{M_0}(\ell')=R_{M_0}(\ell)$,
we have $m_{\ell',j}=m_{\ell,j}=0$ for all $j\in J$. Thus,
$\ell'\notin\bigcup_{j\in J}C_{M_0}(j)$.

For each $\ell'\in[k]$ such that $R_{M_0}(\ell')\neq
R_{M_0}(\ell)$, since $|R_{M_0}(\ell)|=w_{\text{min}}(M_0)$, then
$R_{M_0}(\ell')\nsubseteq R_{M_0}(\ell)$. Note that
$J=[t]\backslash R_{M_0}(\ell)$. Then $R_{M_0}(\ell')\cap
J\neq\emptyset$ and $m_{\ell',j}\neq 0$ for some $j\in J$. So
$\ell'\in C_{M_0}(j)$ and $\ell'\in\bigcup_{j\in J}C_{M_0}(j)$.

Thus, for each $\ell'\in[k]$, $\ell'\notin\bigcup_{j\in
J}C_{M_0}(j)$ if and only if $R_{M_0}(\ell)=R_{M_0}(\ell)$. So
$|\bigcup_{j\in J}C_{M_0}(j)|=k-\Gamma_M(\ell)$.

Case 2: $J\neq[t]\backslash R_{M_0}(\ell)$ for all $\ell\in[k]$
such that $|R_{M_0}(\ell)|=w_{\min}(M_0)$. Then $|\bigcup_{j\in
J}C_{M_0}(j)|=k$. Otherwise, there is an $\ell'\in[k]$ such that
$\ell'\notin C_{M_0}(j)$ for all $j\in J$, which implies that
$m_{\ell',j}=0$ for all $j\in J$, and hence
$R_{M_0}(\ell')\subseteq[t]\backslash J$. Note that
$|J|=t-w_{\text{min}}(M_0)$. Then
$|R_{M_0}(\ell')|\leq|[t]\backslash J|=t-|J|=w_{\text{min}}(M_0)$.
Thus, $|R_{M_0}(\ell')|=w_{\text{min}}(M)=t-|J|$ and
$J=[t]\backslash R_{M_0}(\ell')$, which contradicts to assumption
on $J$.

By the above discussion, we proved that for each $J\subseteq[n]$
of size $|J|= i_0$, either $|\bigcup_{j\in
J}C_{M_0}(j)|=k-\Gamma_{M_0}(\ell)$ for some $\ell\in[k]$ with
$|R_{M_0}(\ell)|=w_{\min}(M_0)$ or $|\bigcup_{j\in
J}C_{M_0}(j)|=k$. Thus, by \eqref{M0-clmn-supp} and
\eqref{eq-gama-mat}, $\xi_{M_0}( i_0)=k-\Gamma(M_0)$.
\end{proof}

\begin{proof}[Proof of Claim 2]
From Definition \ref{ext-inc-mx}, we have
\begin{align}\label{cup-spt-cap}
\textstyle\bigcup\limits_{j\in
J}C_{M}(j)=\textstyle\bigcup\limits_{i'\in[t]: J\cap
J_{i'}\neq\emptyset}S_{i'}, ~ \forall J\subseteq[n].\end{align}

\textbf{Firstly}, we prove $\left|\textstyle\bigcup_{j\in
J}C_{M}(j)\right|\geq\xi_{M_0}(i)$ for each $J\subseteq[n]$ of
size $|J|=\ell$.

By \eqref{def-Ji}, we have
$$(i-1)\beta+1\leq |J|\leq i\beta.$$ Note that by \eqref{def-Ji},
$|J_i|=\beta$. Then the number of $i'$ such that $J\cap
J_{i'}\neq\emptyset$ is at least $i$. By \eqref{cup-spt-cap} and
\eqref{M0-clmn-supp}, we have
\begin{align*}
\left|\textstyle\bigcup\limits_{j\in
J}C_{M}(j)\right|&=\left|\textstyle\bigcup\limits_{i'\in[t]: J\cap
J_{i'}\neq\emptyset}S_{i'}\right|\\&=\left|\textstyle\bigcup\limits_{i'\in[t]:
J\cap J_{i'}\neq\emptyset}C_{M_0}(i')\right|\\&\geq\xi_{M_0}(i).
\end{align*}
The second equation holds because by Definition \ref{ext-inc-mx},
for each $i'\in[t]$, $C_{M_0}(i')=S_{i'}$. So by
\eqref{M-clmn-supp}, we have $\xi_{M}(\gamma)\geq\xi_{M_0}(i)$.

\textbf{Secondly}, we prove there exists a $J\subseteq[n]$ of size
$|J|=\ell$ such that $\left|\textstyle\bigcup_{j\in
J}C_{M}(j)\right|=\xi_{M_0}(i)$.

By \eqref{M0-clmn-supp}, there is a
$\{j_1,\cdots,j_i\}\subseteq[t]$ such that
\begin{align}\label{eq3-dlt-gdc-gm}
\xi_{M_0}(i)=\left|\textstyle\bigcup_{\lambda=1}^iC_{M_0}(j_\lambda)\right|=
\left|\textstyle\bigcup_{\lambda=1}^{i}S_{j_\lambda}\right|.\end{align}
Since $\ell\in J_i$, then by (\ref{def-Ji}),
$\left|\textstyle\bigcup_{\lambda=1}^{i-1}J_{j_\lambda}\right|
=(i-1)\beta<\ell\leq
\left|\textstyle\bigcup_{\lambda=1}^{i}J_{j_\lambda}\right|=i\beta$.
So we can always find a subset $J\subseteq[n]$ such that
$\textstyle\bigcup_{\lambda=1}^{i-1}J_{j_\lambda}\subsetneq
J\subseteq \textstyle\bigcup_{\lambda=1}^{i}J_{j_\lambda}$ and
$|J|=\ell$. Then by \eqref{cup-spt-cap} and
\eqref{eq3-dlt-gdc-gm}, we have
\begin{align*}
\left|\textstyle\bigcup\limits_{j\in J}C_{M}(j)\right|
&=\left|\textstyle\bigcup\limits_{i'\in[t]: J\cap
J_{i'}\neq\emptyset}S_{i'}\right|\\
&=\left|\textstyle\bigcup_{\lambda=1}^{i}S_{j_\lambda}\right|\\&=\xi_{M_0}(i).
\end{align*}

Above discussion implies that
$\xi_{M_0}(i)=\min\limits_{J\subseteq[n],|J|=\ell}
\left|\textstyle\bigcup\limits_{j\in J}C_{M}(j)\right|$. By
\eqref{M-clmn-supp}, we have $\xi_{M}(\ell)=\xi_{M_0}(i)$.
\end{proof}

\begin{proof}[Proof of Claim 3]
We first prove \begin{align}\label{eq1-claim-3}
i\beta-\xi_{M}(i\beta)\leq  i_0\beta-\xi_{M}( i_0\beta), ~ \forall
i\in[i_0].\end{align}

For each $i\in\{1,2,\cdots,t-1\}$, by \eqref{M0-clmn-supp}, there
exists a $J'\subseteq[t]$ of size $|J'|=i$ such that
$$\xi_{M_0}(i)=\left|\textstyle\bigcup_{j\in J'}
C_{M_0}(j)\right|.$$ Pick a $j_0\in[t]\backslash J'$ and let
$J=J'\cup\{j_0\}$. Then by \eqref{M0-clmn-supp},
$$\xi_{M_0}(i+1)\leq\left|\textstyle\bigcup_{j\in J}
C_{M_0}(j)\right|.$$ Above two equations imply that
\begin{align*}
\xi_{M_0}(i+1)-\xi_{M_0}(i)&\leq\left|\textstyle\bigcup_{j\in J}
C_{M_0}(j)\right|-\left|\textstyle\bigcup_{j\in J'}
C_{M_0}(j)\right|\\&\leq|C_{M_0}(j_0)|=|S_{j_0}|=\alpha\\&\leq\beta.
\end{align*}
Combining this with Claim 2, we have
\begin{align*}
i\beta-\xi_{M}(i\beta)&=i\beta-\xi_{M_0}(i)\\
&\leq(i+1)\beta-\xi_{M_0}((i+1)\beta)\\
&=(i+1)\beta-\xi_{M}((i+1)\beta).
\end{align*}
By induction, we have $$\beta-\xi_{M}(\beta)\leq
2\beta-\xi_{M}(2\beta)\leq\cdots\leq
 i_0\beta-\xi_{M}( i_0\beta),$$ which proves
\eqref{eq1-claim-3}.

Now, we can prove Claim 3. Given $i\in[i_0]$ and $\ell'\in J_i$.
Since by \eqref{def-Ji}, $(i-1)\beta+1\leq\ell'\leq i\beta$, and
by Claim 2, $\xi_{M}(\ell')=\xi_{M_0}(i)=\xi_{M}(i\beta)$, then
$$\ell'-\xi_{M}(\ell')\leq i\beta-\xi_{M}(i\beta).$$ Combining this with
\eqref{eq1-claim-3}, we have
$$\ell'-\xi_{M}(\ell')\leq i_0\beta-\xi_{M}(i_0\beta).$$
Note that by \eqref{def-Ji}, $[i_0\beta]=\{1,2,\cdots,i_0\beta\}=
J_1\cup J_2\cup\cdots\cup J_{i_0}$. Thus,
$\ell'-\xi_{M}(\ell')\leq i_0\beta-\xi_{M}( i_0\beta), ~
\forall\ell'\in[i_0\beta]\}.$
\end{proof}

\begin{proof}[Proof of Claim 4]
Denote $\delta'_0= i_0\beta-\xi_{M_0}( i_0)$. We need to prove
$\delta_0=\delta_0'$. Since by Claim 2, $\xi_{M}(
i_0\beta)=\xi_{M_0}( i_0)$, then we have $\delta'_0=
i_0\beta-\xi_{M}( i_0\beta)$.

\textbf{Firstly}, we prove $\xi_M(\ell+\delta_0')\geq\ell$ for all
$\ell\in[k].$

Suppose $\ell\in[k]$. If $\ell+\delta_0'\geq i_0\beta+1$, then by
\eqref{def-Ji}, $\ell+\delta_0'\in J_i$ for some $i\in\{
i_0+1,\cdots,t\}$. By Claim 1 and 2,
$\xi_M(\ell+\delta_0')=\xi_{M_0}(i)=k\geq\ell, \forall
\ell\in[k]$. If $\ell+\delta_0'\leq i_0\beta$, by Claim 3,
$(\ell+\delta_0')-\xi_{M}(\ell+\delta_0') \leq
i_0\beta-\xi_{M}(i_0\beta)=\delta_0'$. So
$\xi_M(\ell+\delta_0')\geq\ell$. Thus,
$\xi_M(\ell+\delta_0')\geq\ell$ for all $\ell\in[k].$

\textbf{Secondly}, we prove that if $\delta'<\delta_0'$, then
$\xi_M(\ell+\delta_0')<\ell$ for some $\ell\in[k]$. We can prove
this by contradiction.

Suppose $\xi_M(\ell+\delta_0')\geq\ell$ for all $\ell\in[k]$. We
have the following two cases:

Case 1: $ i_0\beta-\delta'\in[k]$. Note that $ i_0\beta-\xi_{M}(
i_0\beta)=\delta'_0>\delta'$. Then $\xi_{M}( i_0\beta)<
i_0\beta-\delta'$. Let $\ell=i_0\beta-\delta'$. Then $\ell\in[k]$
and $\xi_{M}(\ell+\delta')<\ell$, which contradicts to assumption.

Case 2: $ i_0\beta-\delta'\notin[k]$. Since
$i_0\beta-\xi_{M}(i_0\beta)=\delta'_0>\delta'$, then
$i_0\beta-\delta'> i_0\beta-\delta'_0=\xi_{M}( i_0\beta)>0$. So we
have $i_0\beta-\delta'>k$ and $i_0\beta>k+\delta'$. By
\eqref{M-clmn-supp} and assumption, we have
$$\xi_{M}( i_0\beta)\geq\xi_{M}(k+\delta')\geq k.$$ By
Claim 2, $\xi_{M}( i_0\beta)=\xi_{M_0}( i_0)$. Then above equation
implies $\xi_{M}( i_0\beta)=\xi_{M_0}( i_0)\geq k$, which
contradicts to Claim 1.

In both cases, we can derive a contradiction. Thus, we conclude
that $\xi_M(\ell+\delta_0')<\ell$ for some $\ell\in[k]$.

Above discussion shows that $\delta_0'$ is the smallest number
that satisfies the condition that $\xi_M(\ell+\delta_0')\geq\ell,
~ \forall \ell\in[k]$.

\textbf{Thirdly}, we prove $\delta_0'\leq n-k$.

Let $J_0$ and $\mathcal G_{J_0}$ be constructed as in the proof of
Lemma \ref{gdc-lem} $($We can denote
$J_0=\{j_1,j_2,\cdots,j_k\}.)$. Then $\mathcal G_{J_0}$ has a
perfect matching. Thus, there exists a permutation
$(i_1,i_2,\cdots,i_k)$ of $(1,2,\cdots,k)$ such that
$m_{i_\lambda,j_\lambda}=1$ for all $\lambda\in [k]$ and we have
$|\textstyle\bigcup_{j\in J'}C_M(j)|\geq|J'|$ for all $J'\subseteq
J_0$. Now, for any $\ell\in[k]$ and $J\subseteq[n]$ of size
$|J|=\ell+n-k$, since $|J_0|=k$, we have $|J\cap J_0|\geq\ell$. So
$|\textstyle\bigcup_{j\in J}C_M(j)|\geq|\textstyle\bigcup_{j\in
J\cap J_0}C_M(j)|\geq|J\cap J_0|\geq\ell$. By \eqref{M-clmn-supp},
we have $\xi_M(\ell+n-k)\geq\ell$. Thus, we proved that
$\delta'=n-k$ also satisfies the condition that
$\xi_M(\ell+\delta')\geq\ell, ~ \forall \ell\in[k]$.

Note that $\delta_0'$ is the smallest number that satisfies the
condition that $\xi_M(\ell+\delta_0')\geq\ell, ~ \forall
\ell\in[k]$. So $\delta_0'\leq n-k$.

\textbf{Finally}, we prove $\delta_0'\geq 0$.

By \eqref{M0-clmn-supp}, there exists a $J\subseteq[t]$ such that
$|J|= i_0$ and $\xi_{M_0}( i_0)=
\left|\textstyle\bigcup\limits_{j\in
J}C_{M_0}(j)\right|\leq\textstyle\sum\limits_{j\in
J}|C_{M_0}(j)|=\textstyle\sum\limits_{j\in J}|S_j|= i_0\alpha\leq
 i_0\beta$. So by Claim 2,
$i_0\beta-\xi_{M}( i_0\beta)= i_0\beta-\xi_{M_0}( i_0)\geq 0.$

Thus, we proved that $0\leq\delta_0'\leq n-k$ and $\delta_0'$ is
the smallest number that satisfies the condition that
$\xi_M(\ell+\delta_0')\geq\ell, ~ \forall \ell\in[k]$. By
\eqref{dlt-min-ell}, we have $\delta_0=\delta_0'=
i_0\beta-\xi_{M_0}( i_0)$.
\end{proof}

\section{Proof of Lemma \ref{mat-contr}}

\begin{proof}[Proof of Lemma \ref{mat-contr}]
Since $k-r\leq{t\choose s}$, we can construct a $k\times t$ binary
matrix $M_0=(m_{i,j})$ such that: 1) $R_{M_0}(i)$,
$i=1,\cdots,k-r$, are mutually different and $|R_{M_0}(i)|=s$; 2)
$|R_{M_0}(i)|=s+1, i=k-r+1,\cdots,k$. Since $t\alpha=sk+r$ and
$0\leq r\leq k-1$, the total number of $1$s in $M_0$ is
$$N_{\text{one}}=(k-r)s+r(s+1)=ks+r=t\alpha.$$ Clearly, $M_0$
satisfies condition (ii). We can further modify $M_0$ properly so
that it satisfies conditions (i) and (ii).

Suppose there is a $j_1\in[t]$ such that $|C_{M_0}(j_1)|<\alpha$.
Since the total number of ones in $M$ is $N_{\text{one}}=t\alpha$,
there exists a $j_2\in[t]$ such that $|C_{M_0}(j_2)|>\alpha$. We
shall modify $M_0$ so that $|C_{M_0}(j_1)|$ increases by one and
$|C_{M_0}(j_2)|$ decreases by one. To do this, let
$$I_1=\{i; 1\leq i\leq k-r, m_{i,j_1}=1 \text{~and~} m_{i,j_2}=0\}$$
and $$I_2=\{i; 1\leq i\leq k-r, m_{i,j_1}=0 \text{~and~}
m_{i,j_2}=1\}.$$ Then clearly, $I_1\cap I_2=\emptyset$ and
$m_{i,j_1}=m_{i,j_2}$ for all
$i\in\{1,\cdots,k-r\}\backslash(I_1\cup I_2)$. We have the
following two cases:

Case 1: There is an $i\in\{k-r+1,\cdots,k\}$ such that
$m_{i,j_1}=0, m_{i,j_2}=1$ and $|R_{M_0}(i)|=s+1$. Then we modify
$M$ by letting $m_{i,j_1}=1, m_{i,j_2}=0$. Then $|C_{M_0}(j_1)|$
increases by one and $|C_{M_0}(j_2)|$ decreases by one. Moreover,
it is easy to see that $M_0$ still satisfies condition (ii).

Case 2: For all $i\in\{k-r+1,\cdots,k\}$, $m_{i,j_2}=1$ implies
$m_{i,j_1}=1$. Note that $|C_{M_0}(j_1)|<\alpha<|C_{M_0}(j_2)|$,
then we have $|I_1|<|I_2|$. For each $\ell\in I_2$, we modify
$M_0$ by letting $m_{\ell,j_1}=1$, $m_{\ell,j_2}=0$ and the other
entries of $M_0$ remain unchanged. Denote the resulted matrix by
$M_\ell$. Then $|C_{M_\ell}(j_1)|$ increases by one and
$|C_{M_\ell}(j_2)|$ decreases by one. If there is an $\ell\in I_2$
such that $M_\ell$ does not satisfy condition (ii), it must be
that $R_{M_\ell}(\ell)=R_{M_\ell}(\ell')$ for some $\ell'\in I_1$.
Moreover, all such $\ell$s and $\ell'$s are in one to one
correspondence. Note that $|I_1|<|I_2|$. Then there exists an $
i_0\in I_2$ such that $M_{ i_0}$ satisfies condition (ii). So we
can let $M_0$ be $M_{ i_0}$.

We can perform the above operation continuously until each column
of $M_0$ has weight $\alpha$. Thus, we obtain a matrix $M_0$ that
satisfies conditions (i) and (ii).
\end{proof}

\end{document}